\documentclass[aps,prl, twocolumn,10pt]{revtex4-1}
\usepackage[utf8]{inputenc}
\usepackage{graphicx}
\usepackage{dcolumn}
\usepackage{bm}
\usepackage{multirow}
\usepackage{amssymb}
\usepackage{amsfonts}
\usepackage{amsmath}
\usepackage{amsthm}
\usepackage{enumerate}
\usepackage{color}
\usepackage{marvosym}
\newcommand{\mathbbm}[1]{\text{\usefont{U}{bbm}{m}{n}#1}}
\usepackage{qcircuit}

\usepackage{tikz}
\usetikzlibrary{matrix,positioning,decorations.pathreplacing}

\newtheorem{theorem}{Theorem}

\newtheorem{lemma}{Lemma}

\newcommand{\id}{\ensuremath{\mathbbm{1}}} 
\newcommand{\one}{\id}

\newcommand{\bra}[1]{\langle #1|}
\newcommand{\ket}[1]{|#1\rangle}
\newcommand{\braket}[2]{\langle #1|#2\rangle}

\newcommand{\C}{\ensuremath{\mathbbm C}}
\newcommand{\be}{\begin{equation}}
\newcommand{\ee}{\end{equation}}
\newcommand{\bea}{\begin{eqnarray}}
\newcommand{\eea}{\end{eqnarray}}

\newcommand{\bi}{\begin{itemize}}
\newcommand{\ei}{\end{itemize}}

\newcommand{\kommentar}[1]{}

\newcommand{\identity}{\mathbbm{1}}

\newcommand{\forget}[1]{}

\begin{document}

\title{All pure fermionic non--Gaussian states are magic states for matchgate computations}


\author{M. Hebenstreit$^1$, R. Jozsa$^2$, B. Kraus$^1$, S. Strelchuk$^2$  and M. Yoganathan$^2$}
\affiliation{$^1$Institute for Theoretical Physics, University of Innsbruck, Technikerstr. 21A, 6020 Innsbruck, Austria\\
$^2$DAMTP, University of Cambridge, Cambridge CB3 0WA, UK}

\begin{abstract}
Magic states were introduced in the context of Clifford circuits as a resource that elevates classically simulatable computations to quantum universal capability, while maintaining the same gate set. Here we study magic states in the context of  matchgate (MG) circuits, where the notion becomes more subtle, as MGs are subject to locality constraints. Nevertheless a similar picture of gate-gadget constructions applies, and we show that every pure fermionic state which is non--Gaussian, i.e. which cannot be generated by MGs from a computational basis state, is a magic state for MG computations. This result has significance for prospective quantum computing implementation in view of the fact that MG circuit evolutions coincide with the quantum physical evolution of non-interacting fermions.
\end{abstract}

\maketitle

{\it Introduction}:
Exploring the landscape intermediate between classical and quantum computing is one of the most interesting issues in quantum information science for both theory and potential implementational impact. It provides the natural context
for the consideration of  novel trade-off possibilities between the individual constituents of the respective theories, that may then provide new applications for emerging near-term quantum hardware of likely limited quantum capability. One fruitful approach in this direction is to determine the classical simulation complexity of a restricted class of quantum processes (that may perhaps enjoy some implementational benefit), and then identify minimal extra resources that would suffice to regain full universal quantum computing power.

The theory of fermionic linear optics underpins a large number of important physical systems, such as Gaussian communication channels~\cite{Br05Lagrange} and important phenomena in condensed matter physics~\cite{zha12, bot04}, including Majorana fermions in quantum wires~\cite{kit01} and Kitaev’s honeycomb lattice model~\cite{kit06}. Together with the development of experimental systems such as cold atoms~\cite{gio08}, atoms in optical lattices~\cite{jor08}, and quantum dots~\cite{los98,han07} it is well-suited for a range of information processing tasks.

Early on, it was shown that the computational capabilities of unassisted fermionic linear optics can be described by matchgate circuits (MGCs) and they are entirely classically efficiently simulatable~\cite{TeDi02, Knill, JoMi08, JoMi15}, the latter holding also for some extensions of fermionic linear optics with dissipative processes~\cite{kb11}. 
Another class of quantum processes that is classically simulatable is given by Clifford circuits, albeit for different reasons in comparison to MGCs ~\cite{joz08} and indeed  the classically simulatable computational power of MGCs has the neat characterisation of being equivalent to log-space bounded universal unitary quantum computation~\cite{JoKr10}.

Determining extra ingredients for Clifford circuits \cite{BrKi05,BrSm16} or for MGCs that suffice to regain universal computational power, can give avenues for testing and boosting the power of corresponding near-term quantum computing devices that are based on implementing such gates. In the case of Clifford circuits a fundamentally important such ingredient is the provision of a so-called magic state \cite{BrKi05} i.e. a suitably chosen additional input state whose availability together with adaptive measurements gives universal computing power while still using only the same gate set, via introduction of an associated ``gate-gadget" construction.

In this Letter we establish and study the notion of magic states for MGCs.  We will see that this notion becomes more subtle in the matchgate context, as matchgate actions are subject to a locality constraint, and also the SWAP gate is not available to freely move magic states into arbitrary positions amongst the qubits. 
Nevertheless a similar picture of gate-gadget constructions applies, and our main result will be to show that {\em every}  pure fermionic state which is non--Gaussian, i.e. which cannot be generated by a MGC from a computational basis state, is a magic state for MGCs. Along the way we will see that the matchgate locality constraint imposes a necessary condition that magic states be fermionic states; and we will give an explicit example of a gate-gadget construction (with associated 4-qubit magic state) for implementing the SWAP gate, which is known to extend the power of MGCs to full quantum universal power \cite{JoMi08}.

{\it Preliminaries}: The Pauli operators are denoted by $X,Y,Z$. For $n$ qubits, the even (resp. odd) parity subspace, is the eigenspace of the $n$-qubit operation $Z^{\otimes n}$ with eigenvalues $1$ ($-1$), respectively. An operator is called {\em even} if it preserves the even and odd parity subspaces.
A matchgate (MG) is a two--qubit unitary even operator $G(A,B)=A\oplus B$ where $A, B \in U(2)$ act on the even and odd parity subspaces respectively and satisfy $\operatorname{det}A = \operatorname{det}B$. Whereas the fermionic SWAP gate, $fSWAP=G(Z,X)$, is a MG, the $SWAP$ gate, $SWAP=G(\one,X)$ is not. A matchgate circuit (MGC) is a quantum circuit which comprises MGs acting only on {\em nearest neighbor} (n.n.) lines. Correspondingly, in the
following the term MG will always refer to a nearest neighbor matchgate. The action of any MGC is always an even operator.
It has been shown \cite{Va02,TeDi02, JoMi08} that for any MGC, if the input is a computational basis state and the output is a final computational basis measurement on any single qubit, then the output is classically efficiently simulatable. Moreover, in \cite{JoKr10} it has been shown that MGCs running on $n$ qubits can be compressed into a universal quantum computer running on ${\cal O}(\log(n))$ qubits. However, supplementing MGCs with additional gates, such as the SWAP-gate, makes the circuits universal for quantum computing \cite{JoMi08, BrGa11, BrGa12,BrCh14}. This is analogous to the situation of classically simulatable Clifford computations being elevated to universal quantum computing power by the inclusion of a non-Clifford gate such as the $T$ gate.

MGCs have a profound physical significance in that they correspond exactly to quantum evolutions of so-called non-interacting fermions \cite{TeDi02,Knill}. This correspondence is given explicitly by the Jordan-Wigner (JW) transformation as we now summarise. Introducing an annihilation and a creation operator $b_k,b_k^\dagger$ for each fermionic mode $k=1,\dots,n$ and writing $\tilde{c}_{2k-1}=b_k+b_k^\dagger, \tilde{c}_{2k}=-i(b_k-b_k^\dagger)$, the canonical anti-commutation relations can be algebraically represented by the $n$-qubit JW operators
 \begin{align}
 \label{eq:JWT1} c_{2j-1}&=Z\otimes Z\otimes \ldots \otimes Z \otimes  X_j \otimes \one \otimes \ldots \otimes \identity \\ \nonumber
c_{2j}&=Z\otimes Z\otimes \ldots \otimes Z \otimes Y_j \otimes \one \otimes \ldots \otimes \identity 
\end{align}
and then the JW transformation is the unitary mapping between the Fock space
 of $n$ fermionic modes and the Hilbert space of $n$ qubits whereby a general fermionic state
 \begin{align} \label{eq_FermiState}
\ket{\Psi}=\sum_{i_1\ldots,i_n\in{0,1}} \alpha_{i_1\ldots,i_n}
(b_1^\dagger )^{i_1}(b_2^\dagger )^{i_2}\ldots (b_n^\dagger
)^{i_n}\ket{\Omega},
\end{align}
 (where $\ket{\Omega}$ denotes the vacuum state and
$\alpha_{i_1\ldots,i_n}\in \C$) is mapped to the $n$-qubit state
\begin{align} \label{eq_FermiState2} \ket{\Psi}=\sum_{i_1\ldots,i_n} \alpha_{i_1\ldots,i_n} \ket{i_1\ldots i_n}.
\end{align}
Correspondingly (in view of the boson-fermion superselection rule) an $n$ qubit state $\ket{\Psi}$ is called {\em fermionic} if it is an eigenstate of $Z^{\otimes n}$.

We will also use the following terminology frequently. An $n$-qubit operation $W$ is called {\em Gaussian} if it arises as the action of a MGC (which in turn always corresponds via the JW transformation, to evolution under a quadratic fermionic Hamiltonian). An $n$-qubit state $\ket{\Psi}$ is called {\em Gaussian} if it arises as the action of a Gaussian operation on a computational basis state. In \cite{Br05Lagrange,SpSc18} it has been shown that a fermionic state is a Gaussian state iff $\Lambda_n \ket{\Psi}^{\otimes 2}=0$, where $\Lambda_n=\sum_{i=1}^{2n} c_i \otimes c_i$. Furthermore, any even operator $R$ is Gaussian iff $[\Lambda_n,R^{\otimes 2}]=0$ \cite{Br05Lagrange}. 

As MGs can be freely applied in circuits without altering the classical simulability of the computational output, we call a product of MGs, i.e. a Gaussian unitary operator, a {\em free} or {\em resourceless} operation, and introduce a {\em resourceful} gate as one which leads to a universal gate set if used in conjunction with free operations. Here quantum computational universality may occur in an encoded sense, e.g. as in \cite{JoMi08} where it is shown that n.n. $SWAP$ is a resourceful gate.

We call two states $\ket{\phi_1}$ and $\ket{\phi_2}$ {\em MG--equivalent} if there exists a free operation which transforms $\ket{\phi_1}\ket{b_1}$ into $\ket{\phi_2}\ket{b_2}$ where $\ket{b_1},\ket{b_2}$ are computational basis states. Clearly MG-equivalence is indeed an equivalence relation. Free (i.e. Gaussian) operations are always even operators and the inclusion of computational basis state ancillas here (which we can think of as ``free'' states) allows our notion of MG-equivalence to usefully extend to apply between some even and odd fermionic states e.g. $\ket{\phi_1}= \ket{00}+\ket{11}$ and $\ket{\phi_2}=\ket{01}+\ket{10}$ being MG-equivalent via use of ancillas $\ket{b_1}=\ket{0}, \ket{b_2}=\ket{1}$ and the free operation $G(X,X)$ applied on the last two qubits.

{\it Magic states}: The notion of magic state has been introduced in the context of Clifford circuits \cite{BrKi05} which then comprise the free operations. Adding the $T$ gate makes the Clifford gate set universal. However, instead of enlarging the gate set, one can alternatively consider allowing more general input states, so-called magic states, and adaptive measurements. In this way one can realize a $T$ gate via the so-called $T$-gadget \cite{BrKi05}, that consumes one copy of the magic state $\ket{T} = 1/\sqrt{2} \left(\ket{0} + e^{i\pi/4} \ket{1}\right)$ and uses one adaptive measurement in the computational basis together with only Clifford gates. Let us stress here that---considering computations with a single output mesaurement---neither copies of the magic state, nor adaptive measurements in the computational basis by themselves give rise to universal quantum computation (assuming quantum is more powerful than classical computing). Indeed these situations are classically efficiently simulatable \cite{Go99, AaGo04, ClJo08, Va10, JoVa14}. However, the combination of these two ingredients is resourceful. As we will be concerned with adaptive measurements in the computational basis only, we will from now on simply call them adaptive measurements.

In generality, we introduce the following natural definition:  if $R$ is a resourceful $k$-qubit gate for a set of free operations, we say that an $m$-qubit state $\ket{M}$ is a {\em magic state} for $R$ if\\
 (M1): there is a circuit $C$ of free gates and adaptive measurements such that for any $k$-qubit state $\ket{\alpha}$, $C$ maps $\ket{\alpha}\ket{M}$ to $(R\ket{\alpha})\ket{\tilde{M}}$ (where $\ket{\tilde{M}}$ is any state, that may depend on the intermediate measurement outcomes too. However, it may not depend on $\ket{\alpha}$.)\\
 Actually we will need a slightly more general version of (M1) as follows. We will require that for any $\epsilon>0$,   $C$  (of circuit size $O (poly(1/\epsilon))$) acting on $\ket{\alpha}$ together with $O(poly(1/\epsilon))$ copies of $\ket{M}$,  produces $R\ket{\alpha}$ with probability $1-\epsilon$ (where the probabilistic randomness here arises from the intermediate measurement outcomes.)  For bounded error computations with input size $n$ we take $\epsilon = 1/poly(n)$, which then maintains efficiency of the computation.

Note that there are key differences that will require special care when considering MG computation compared to Clifford computation, necessitating a further condition (M2) as below. For Cliffords, $SWAP$ is a free operation, so any free (multi-qubit) gate can be placed to act on any (distant) lines, and magic states can be freely moved to any position amongst the qubit lines when required or else placed there at the outset in the initial input state, all without affecting any action of free gates. But for MGCs none of these features hold!  - MGs can act on n.n. lines only and $SWAP$ is not a free gate so states cannot generally be freely moved around amongst the qubit lines. In particular magic states cannot be freely moved into the positions needed for their use in implementing a resourceful gate; and nor can they be placed in their needed positions from the start, as this would partition the circuit lines into sectors that must then remain independent under processing by n.n. gates, at least until the magic state has been suitably disposed of. In view of these features, in addition to (M1) we impose a second condition (M2) on a state $\ket{M}$ for it to be a magic state:\\
(M2): The state $\ket{M}$ can be swapped through arbitrary states via use of free gates only.

Thus as for Clifford circuits, magic states can be prepared prior to the computation and can then be used whenever and wherever needed.

Note that (M2) trivially holds in the Clifford scenario, but for MGCs it is a non-trivial condition and it is not even a priori clear that there exist {\em any} states satisfying both (M1) and (M2). However, an example of such a state has been identified in \cite{Br06} in the context of topological quantum computation.
 Below we will show that (M2) requires $\ket{M}$ to be a fermionic state (i.e. have definite parity) and our main result is that, in fact, {\em every} fermionic state that is not Gaussian is a magic state for MG computations. Note also that in (M2) we require swapping through arbitrary states rather than just fermionic ones, as we will generally wish to move $\ket{M}$ across some but not all, lines of an ambient fermionic state and that subset of qubit lines will not generally have an associated definite parity for the superposition components.

{\it Magic states for MG computations}: To begin our characterisation, we note that magic states for matchgate circuits cannot be single qubit states nor products of single qubit states. This follows from the fact that matchgate circuits remain classically simulatable even if arbitrary product input states and adaptive measurements in the computational basis are allowed \cite{Br16}. In this vein it is instructive to consider the state $\ket{+}=1/\sqrt{2} (\ket{0}+\ket{1})$ which can be used to implement the resourceful Hadamard gate $H$ using only matchgates \cite{Br16}. 
 Hence, one might conclude that supplementing a matchgate circuit with copies of $\ket{+}$ allows universal quantum computation. However, the reason why $\ket{+}$ in fact does not give universal power is that it fails to satisfy (M2), and is debilitatingly subject to the issues in our discussion above, motivating the introduction of (M2) \cite{footnote1}.

\begin{lemma}
\label{obs:swapping}
A multi-qubit state $\ket{\psi}$ on adjacent lines can be swapped through a neighboring line in an arbitrary state using free gates iff it is fermionic.
\end{lemma}
The proof (including specification of allowable swapping processes) is given in the Supplemental Material \cite{supp}. Note that the fact that any magic state must be fermionic then implies that the resourceful gate which can be implemented utilizing this magic state must be fermionic too, i.e. map fermionic states to fermionic states. This follows from the fact that if we start with a fermionic state and apply (i) any Gaussian unitary (i.e. any sequence of MGs) and (ii) any subsequent measurement in the computational basis, then the result is a fermionic state too. Note further that any Gaussian state can be generated via a MGC so it can never be a magic state. Thus, any magic state must be a fermionic state which is non--Gaussian. The main result of this paper is to show that in fact {\em any} fermionic state which is non--Gaussian is a magic state for MG computations (see Theorem 1). 

Before presenting our main result we give a simple explicit example of a 4-qubit magic state, thus also showing that (M1) and (M2) are not mutually inconsistent. Note that the smallest number of qubits that can be used is 4, as all 2- and 3-qubit fermionic states are Gaussian~\cite{Br05}.
The state $\ket{M} = \ket{\phi^+}_{13}\ket{\phi^+}_{24} = 1/2 \left(\ket{0000}+\ket{0101}+\ket{1010}+\ket{1111}\right)_{1234}$ is a magic state. For the qubit ordering and partition $14|23$ it is the Choi state corresponding to the $SWAP$ gate. Hence one copy of $\ket{M}$ can deterministically implement the resourceful $SWAP$ gate using MGs and adaptive measurements by gate teleportation \cite{GoCh99}, as shown in Figure~\ref{fig:magic4}. After teleporting the input with Bell measurements, the swapped input states are available on lines 2 and 3 up to local Pauli corrections, which are corrected right after the measurements. For $Z$ corrections we use the matchgate $Z\otimes I = G(Z,Z)$. For $X$ corrections we first use $G(Z,X)$'s to move in a $\ket{0}$ ancilla, then apply $G(X,X)= X\otimes X$ to it and the line, and finally use $G(-Z,X)$'s to remove the ancilla, now in state $\ket{1}$
\cite{footnote2}. 
Note further that the required Bell measurements can be realized by two local computational basis measurements preceded by the MG $G(H,H)$ as also shown in Figure~\ref{fig:magic4}.  The four qubit lines in computational basis states resulting from the Bell measurements are also moved out using $G(\pm Z,X)$'s (not shown in the Figure).
$\ket{M}$ hence fulfills (M1) for $R=SWAP$. Moreover, $\ket{M}$ is a fermionic state and hence (by Lemma~\ref{obs:swapping}) fulfills  (M2) (see Figure \ref{fig:main}). As in the case of Clifford circuits, considering circuits with single output measurements, neither the magic state itself (and copies thereof), nor the adaptive measurement alone results in a circuit which is no longer classically simulatable \cite{He19b}; it is only the combination of both that makes them resourceful.

\begin{figure}[ht]
    \centering
    \includegraphics[width=1.0\linewidth]{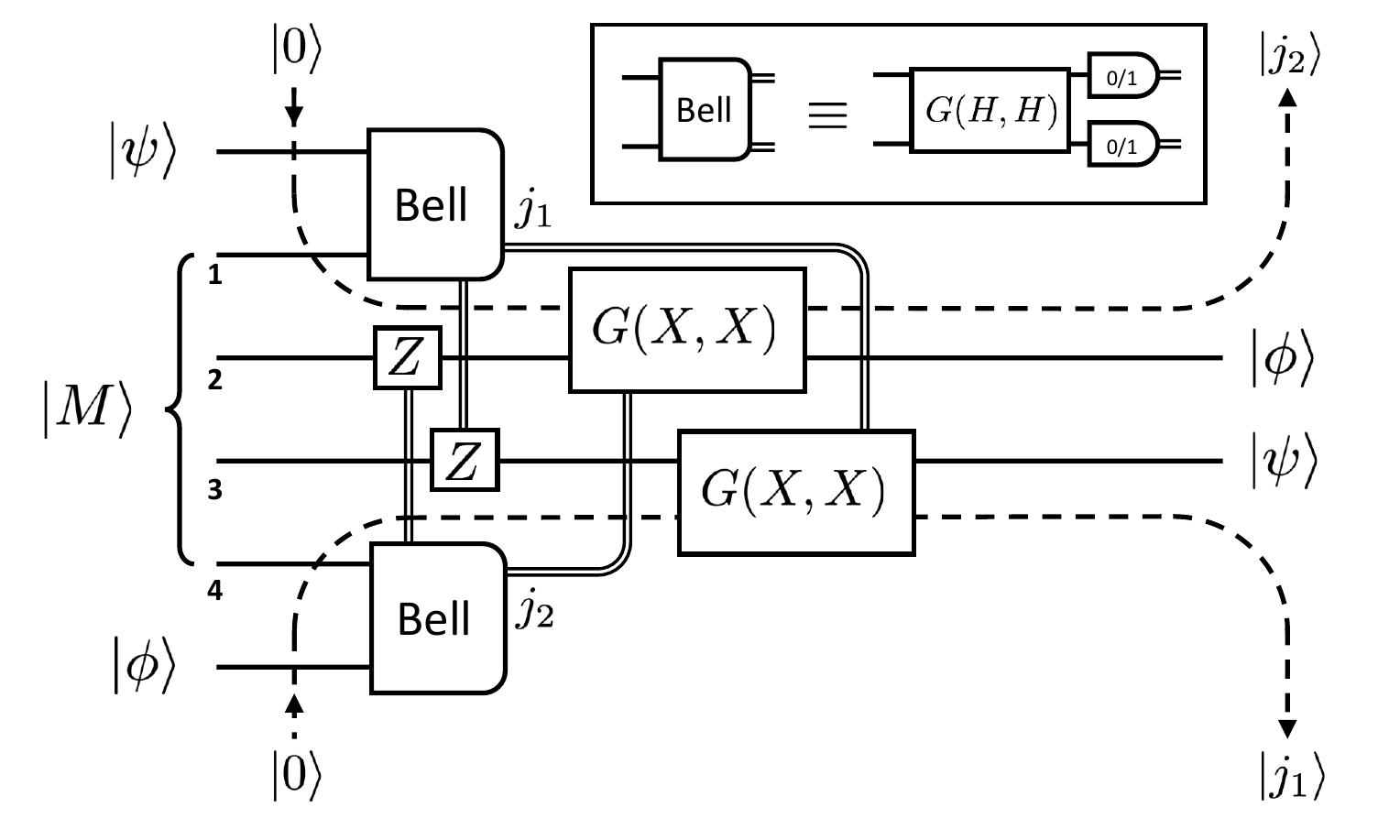}
    \caption{The $SWAP$-gadget deterministically implements a $SWAP$ gate by gate teleportation over the magic state $\ket{M} = \ket{\phi^+}_{13}\ket{\phi^+}_{24}$ as described in the text. Bell measurements can be implemented by $G(H,H)$ followed by single qubit measurements in the computational basis.}
\label{fig:magic4}
\end{figure}

\begin{figure}
  \centering
    \includegraphics[width=0.8\linewidth]{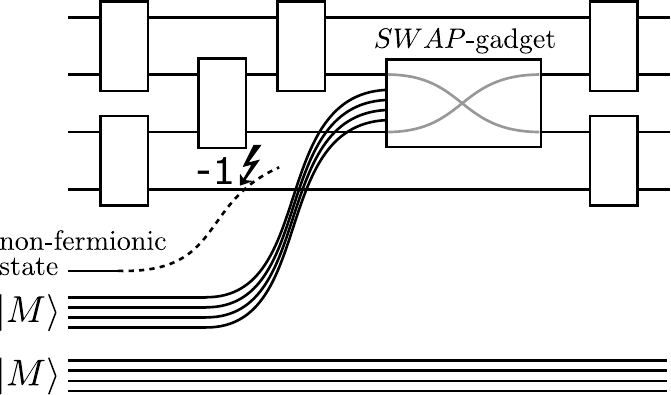}
  \caption{ Non-fermionic states do not satisfy (M2). Swapping them through other lines with e.g. $fSWAP$, a relative phase (-1) is picked up. The state $\ket{M}$ satisfies (M2), thus it can be swapped through arbitrary states via free operations. It moreover satisfies (M1), i.e., it allows implementation of a resourceful operation, here the $SWAP$-gate. As a main result of this article we show that actually all non-Gaussian fermionic states can be swapped through arbitrary states via free operations and moreover allow the implementation of resourceful gates.}
  \label{fig:main}
\end{figure}

Clearly any state which is MG--equivalent to $\ket{M}$ is a magic state too. A particularly useful MG--equivalent state is the $\ket{GHZ_4}$ state. It has the property that the first (last) two qubits can be fermionically swapped to arbitrary positions in the circuits, respectively. Hence, a $SWAP$ gate between distant qubits can be realized by consuming only a single copy of the magic state, instead of ${\cal O}(n)$ copies of the state $\ket{M}$. In these examples the resourceful gate is implemented deterministically consuming one copy of the magic state. Relaxing this condition slightly, as in our statement of (M1), leads to the following lemma, which is proven in the Supplemental Material \cite{supp}.

\begin{lemma}
\label{lemma:4qubitmagicfamily}
$\ket{\psi_{\phi}} = 1/2 (\ket{0000} + \ket{0011} + \ket{1100} + e^{i \phi} \ket{1111})$ is a magic state for matchgate circuits for all $\phi \in (0, 2\pi)$. The resourceful controlled-$\phi$-phase gate can be realized with arbitrary high success probability $1-\epsilon$ ($\epsilon >0$) by consuming ${\cal O}(poly(1/\epsilon))$ copies of the state $\ket{\psi_{\phi}}$.
\end{lemma}

Next we show that any entangled 4--qubit fermionic state which is non--Gaussian is MG--equivalent to $\ket{\psi_{\phi}}$ for some $\phi \in (0, 2\pi)$. Hence any entangled 4--qubit fermionic state which is non--Gaussian is magic. We show this by constructing an explicit MGC of depth three, which transforms any given 4--qubit fermionic state into a state of the form  $\ket{\psi_{\phi}}$. The construction is given in the Supplemental Material \cite{supp}. 

Finally we use the above result to prove the main result of the paper:

\begin{theorem}
\label{Th_allfermi}
Any pure fermionic state which is non--Gaussian is a magic state for matchgate computations.
\end{theorem}

Recall that a Gaussian state can be generated from a computational basis state via a MGC, so it can never be a magic state. Furthermore, we have seen already that (M2) implies that any magic state must be fermionic. Hence Theorem~\ref{Th_allfermi} shows that the largest set of possible states is indeed magic. We outline the idea of the proof of Theorem~\ref{Th_allfermi} here, and full details are given in the Supplemental Material \cite{supp}.
The proof is by induction. As shown above, any $k$-qubit fermionic state for $k=4$ which is non--Gaussian is a magic state. Assume that the statement holds for some $k\geq 4$ and let $\ket{\Psi}$ be a $(k+1)$-qubit fermionic, non-Gaussian state. We show that any such state can be transformed via Gaussian unitaries and measurements into a $k$--partite fermionic, non--Gaussian state. This can be seen as follows. Let $P_b^j$ for $b\in \{0,1\}$ denote the projector onto the state $\ket{b}$ of qubit $j$.
As $\ket{\Psi}$ is non--Gaussian, it holds that $\Lambda_{k+1} \ket{\Psi}^{\otimes 2}\neq 0$. It can then be shown that for any such state there exists a qubit $j$ and a bit value $b\in \{0,1\}$ such that either $P_b^j \ket{\Psi}=\ket{b}\ket{\Psi_{k}}$ is non-Gaussian, or else one of $P_b^{j} G_{j,j+1}(H,H) \ket{\Psi}=\ket{b}\ket{\Phi_{k}}$ or $P_b^{j+1} G_{j,j+1}(H,H) \ket{\Psi}=\ket{b}\ket{\Phi_{k}}$ is non--Gaussian. Hence by choosing the correct option, we obtain a $k$-qubit fermionic, non--Gaussian state i.e. a $k$-qubit magic state, from $\ket{\Psi}$, which proves the statement. It should be noted here that the $k$--partite magic state, obtained as a choice of post-measurement state, is not obtained deterministically. However, the success probability is independent of the width of the circuit in which it is to be subsequently used to implement a resourceful gate.

In summary, we have shown that the set of magic states for MG computation coincides with the set of fermionic non--Gaussian states. Our work raises a natural question about whether one can use this result in the setting where MG circuits can be compressed into exponentially smaller quantum computations~\cite{JoKr10}. For practical reasons it will be necessary to study the usefulness of mixed states as magic states. Whereas certain mixed states remain classically efficiently simulatable \cite{OsGu14}, magic state distillation will be subject of future research~\cite{Br06}. Another interesting direction is to study whether the decomposition of a universal quantum computation into classically simulatable parts and gadgets using magic states might lead to the development of novel verification schemes for quantum processes, much in the spirit of~\cite{JoSt17}.

\begin{acknowledgments}
{\it Acknowledgements:} M.H. and B.K. acknowledge financial support from the Austrian Science Fund (FWF) grant DK-ALM: W1259-N27 and the SFB BeyondC. Furthermore, B.K. acknowledges support of the Austrian Academy of Sciences via the Innovation Fund ``Research, Science and Society''. R.J., S.S. and M.Y. acknowledge support from the QuantERA ERA-NET Cofund in Quantum Technologies implemented within the European Union's Horizon 2020 Programme (QuantAlgo project), and administered through the EPSRC grant EP/R043957/1. M.Y. is supported by the Australia Cambridge Bragg Scholarship scheme, and S.S. by the Leverhulme Early Career Fellowship scheme.
\end{acknowledgments}

\newpage
\section{Supplemental Material}

In this supplemental  material we present the proofs of Lemma~\ref{obs:swapping}, \ref{lemma:4qubitmagicfamily}, and Theorem~\ref{Th_allfermi} from the main text.

\section{Proof of Lemma 1}

We first note two useful facts:\\
(F1): $\ket{0}$ or $\ket{1}$ can be freely swapped to any position within any $k$-qubit state using $fSWAP=G(Z,X)$ or $G(-Z,X)$ respectively.\\
(F2): any two computational basis states on $n$ qubits of the same parity are MG-equivalent (with no ancillas needed) since $G(X,X)$ is $X\otimes X$ so we can change the number of $\ket{1}$'s by 0, $\pm 2$, and we can freely reorder the qubits, by (F1).

\begin{lemma}
\label{obs:swapping}
A multi-qubit state $\ket{\psi}$ on adjacent lines can be swapped through a neighboring line in an arbitrary state using free gates iff it is fermionic.
\end{lemma}
\begin{proof}
For the ``if'' part: suppose that $\ket{\psi}$  of $k$ qubits, is fermionic and let us first consider swapping $\ket{\psi}$ across a line $A$ in pure state $\alpha \ket{0} + \beta \ket{1}$. If $\ket{\psi}$ is even (i.e., supported entirely in the even subspace) then use of $k$ $fSWAP=G(Z,X)$ gates each on line $A$ and the neighboring line of $\ket{\psi}$, will achieve the desired swapping while the amplitude $\beta$ undergoes an even number of sign changes during the process (as $fSWAP$ maps $\ket{1}\ket{1}$ to $-\ket{1}\ket{1}$, $\ket{0}\ket{1}$ to $\ket{1}\ket{0}$, $\ket{1}\ket{0}$ to $\ket{0}\ket{1}$, and $\ket{0}\ket{0}$ to $\ket{0}\ket{0}$), so at the end the final state of $A$ is left unchanged. If $\ket{\psi}$ is odd, then we adjoin an ancilla qubit $\ket{1}$ to produce the even $(k+1)$-qubit fermionic state $\ket{\psi}\ket{1}$ (e.g. by (F1) we can swap $\ket{1}$ in from a fringe position using $G(-Z,X)$ gates) and then proceed as above with $k+1$ $G(Z,X)$ gates. Finally using $G(-Z,X)$'s again, we remove the ancilla $\ket{1}$ back out to a fringe position. For the general case of line $A$ possibly being entangled with other lines, the situation is as above except that now the amplitudes $\alpha$ and $\beta$ are vectors in another Hilbert space, and the process still proceeds just as above, now with the vector $\ket{\beta}$ changing sign an even number of times.

Let us now prove the ``only if'' part: suppose that $\ket{\psi}$ is any $k$ qubit state.
We will see that it suffices to consider the case that the line $A$ over which it is to be swapped, is in a pure state, as this will already require that $\ket{\psi}$ be fermionic. Let us denote the 1-qubit pure state of $A$ by $\ket{A}$.
Note that the operation that accomplishes swapping of $\ket{\psi}$ with $\ket{A}$ need not necessarily be a sequence of $G(Z,X)$'s and $G(-Z,X)$'s  but can be a most general Gaussian unitary $U$, which may depend on $\ket{\psi}$ but not on $\ket{A}$; indeed $U$ must achieve swapping for all choices of $\ket{A}$. Moreover we may also have availability of an ancilla $\ket{b}$ like $\ket{1}$ above for the odd fermionic case. It must be able to be freely moved in so for example, by (F1), $\ket{b}$ could be a computational basis state. However our argument below applies to $\ket{b}$ being any state and we will not need to impose the condition that it can be freely moved.

Let us now assume such a Gaussian unitary $U$ exists. Then it must hold that
\begin{align}\label{A1}
U \ket{A} \ket{\psi}\ket{b} =  \ket{\psi}\ket{A} \ket{\xi} \hspace{3mm} \mbox{for all $\ket{A}$,}
\end{align}
where $\ket{b}$ is an ancilla state of say $m$ qubits, and $\ket{\xi}$ is the final state of these ancilla qubits. Both $\ket{b}$ and $\ket{\xi}$ may depend on $\ket{\psi}$ but not on $\ket{A}$.
It will also be necessary for $\ket{\xi}$ to be able to be moved out without disturbing the ambient state, but again, even without imposing this condition we will see that the above process will already require $\ket{\psi}$ to be fermionic.

To start, introduce
\begin{align} \label{ups0}
\ket{\Upsilon_0}&=\ket{0}\ket{\psi}\ket{b} \text{ and} \\
\label{ups1}
\ket{\Upsilon_1}&=\ket{1}\ket{\psi}\ket{b}.
\end{align}

Then eq. (\ref{A1}) is equivalent to 

\begin{align} \label{A2}
U \ket{\Upsilon_0} &=  \ket{\psi}\ket{0}\ket{\xi} \text{ and} \\
\label{A3} U\ket{\Upsilon_1} &= \ket{\psi}\ket{1}\ket{\xi}.
\end{align}
As $U$ is Gaussian, it must hold that $[\Lambda_{1+k+m}, U \otimes U] = 0$, which is equivalent to
\begin{align}
\sum_{i=1}^{2(1+k+m)}  (U^\dagger \otimes U^\dagger) (c_i \otimes c_i) (U \otimes U)= \sum_{i=1}^{2(1+k+m)} c_i \otimes c_i.
\end{align}


Recall that the JW operators for $1+k+m$ qubits have the product Pauli form $c_l=Z \otimes \ldots \otimes Z \otimes P \otimes \identity \otimes \ldots \otimes \identity$ where $P$ in position $j$ (for $j=1,\ldots ,1+k+m$) is either $X$ (for $l=2j-1$) or $Y$ (for $l=2j$), and crucially here, $Z$ and $I$ have zero off-diagonal entries while $X$ and $Y$ have non-zero entries.

Let us now apply each side of eq. (\theequation) to $\ket{\Upsilon_0}\otimes \ket{\Upsilon_1}$ and take inner product with $\ket{\Upsilon_1}\otimes \ket{\Upsilon_0}$ (i.e. apply the corresponding bra vector from the left to eq. (\theequation)). Note that eqs. (\ref{A2}, \ref{A3})  are then relevant for evaluating the LHS of eq. (\theequation). Because of the off-diagonal properties above and positions of $Z$ and $I$ Pauli's in the JW operators, in the sum on the RHS of Eq. (\theequation),  only terms with $i \in \{1,2\}$ survive, while on the LHS only  terms with $i \in \{2k+1,2k+2\}$ survive, and we get
\begin{align}
 2 \bra{\psi} Z^{\otimes k} \ket{\psi}^2\braket{b}{b}^2 = 2 \braket{\psi}{\psi}^2\braket{\xi}{\xi}^2,
\end{align}
so that  $\bra{\psi} Z^{\otimes k} \ket{\psi} = \pm 1$ i.e. $\ket{\psi}$ must be fermionic.

\end{proof}

\section{Proof of Theorem 1}

 We first show that any state $\ket{\psi_{\phi}} = 1/2 (\ket{0000} + \ket{0011} + \ket{1100} + e^{i \phi} \ket{1111})$ with $\phi \in (0, 2\pi)$ is a magic state (cf Lemma \ref{lemma:4qubitmagicfamily} in the main text). We then show Lemma \ref{lemma:equivalenc4qubits}, which states that any fermionic non--Gaussian 4-qubit state is MG--equivalent to one of the above mentioned states. Finally we prove the main result of the paper, by showing that any fermionic non--Gaussian state can be projected via free operations (i.e. MGs and measurements in the computational basis) onto a fermionic non--Gaussian 4--partite state. 
This projection will be achieved as a post-measurement state of a measurement with success probability being independent of the width of the ambient circuit in which it is applied.

\begin{lemma}
\label{lemma:4qubitmagicfamily}
$\ket{\psi_{\phi}} = 1/2 (\ket{0000} + \ket{0011} + \ket{1100} + e^{i \phi} \ket{1111})$ is a magic state for matchgate circuits for all $\phi \in (0, 2\pi)$. The resourceful controlled-$\phi$-phase gate can be realized with arbitrary high success probability $1-\epsilon$ ($\epsilon >0$) by consuming ${\cal O}(poly(1/\epsilon))$ copies of the state $\ket{\psi_{\phi}}$.
\end{lemma}
\begin{proof}
The state $\ket{\psi_{\phi}}$ is fermionic and thus satisfies (M2) according to Lemma \ref{obs:swapping}. Let us now show that the state moreover satisfies condition (M1). 
To this end, note that a controlled-$\phi$-phase gate $C_\phi$ is indeed resourceful for $\phi \in (0, 2\pi)$ in the sense that together with nearest-neighbor matchgates it forms a universal gate set \cite{BrGa11, OsZi17}.
We will show that $\ket{\psi_{\phi}}$  satisfies (M1) in two steps, reminiscent of an idea used in \cite{CiDu01}. First, we will show that one copy of $\ket{\psi_{\phi}}$ allows to probabilistically implement either $C_\phi$, or $C_{-\phi}$, with respective probabilities of $1/2$, at any desired place in the circuit.
In a second step we show that for all $\epsilon > 0$, $O({\rm poly}(1/\epsilon))$ many copies of $\ket{\psi_{\phi}}$ suffice in order to implement the desired phase gate with a probability of success of at least $1-\epsilon$, as required in (M1). 

Utilizing the same gadget as in Figure~1 in the main text inputting $\ket{\psi_{\phi}}$ instead of $\ket{M}$, one successfully implements $C_\phi$ in case the outcomes of both Bell measurements are $\ket{\phi^+}$, as $\ket{\psi_{\phi}}$ is the Choi–Jamiołkowski state corresponding to $C_\phi$. Although the gate in question is not Clifford unless $\phi=\pi$, let us nevertheless implement the Pauli corrections as in Figure~1. While $Z$ commutes with $C_\phi$, $(X^i \otimes X^j) C_\phi (X^i \otimes X^j)$ permutes the diagonal entries of $C_\phi$. Nevertheless, local $\pm \phi$-phase gates (extended by the identity) are matchgates, and can be adaptively applied in order to successfully implement $C_\phi$ in case an even number of $X$ corrections was necessary. Otherwise the gate $C_{-\phi}$ was realized. All measurement outcomes are equally likely, hence we have implemented $C_\phi$ or $C_{-\phi}$ with probabilities of $1/2$, respectively.

As shown in \cite{CiDu01}, this allows to make the gate implementation successful with an arbitrarily high probability. Indeed, in case $C_{-\phi}$ has been implemented instead of $C_{\phi}$, the gate implementation can be repeated using  $\ket{\psi_{2 \phi}}$. With probability $1/2$, overall the gate $C_\phi$ is implemented. In case of failure, $C_{-3\phi}$ is implemented. Iterating this process $L$ times, $C_\phi$ is implemented with a probability $1-1/2^L$. We consider now a protocol which uses up to $L$ iterations and halts. This procedure, however, not only requires to (probabilistically) implement $C_\phi$, but also controlled-phase gates for multiples of the desired phase $C_{2^k \phi}$ for $k \in \mathbb{N}$. The required states $\ket{\psi_{2^k \phi}}$ can be probabilistically created using copies of $\ket{\psi_{\phi}}$ as explained in the following. Iteratively using a copy of $\ket{\psi_{2^{j-1} \phi}}$ in order to implement $C_{2^{j-1}\phi}$ on another copy of $\ket{\psi_{2^{j-1} \phi}}$ yields $\ket{\psi_{2^{j} \phi}}$ with probability $1/2$. Note that here, in case of failure, there is no necessity to correct the failure, instead one can simply discard the failure states and try again with additional copies of $\ket{\psi_{\phi}}$. As the procedure of generating  $\ket{\psi_{2^{k} \phi}}$ is probabilistic itself, the overall strategy of implementing  $C_\phi$ has two possible sources of failure. First, one might fail in implementing $C_{\phi}$  in any of the $L$ rounds and, second, one might be unlucky in generating the states $\ket{\psi_{2^k\phi}}$  leading to a resource consumption that is too high. 

Let us now, however,  show that the procedure above implements $C_{\phi}$ with a sufficiently large probability and a sufficiently small number of resource states $\ket{\psi_{\phi}}$. To this end, we now analyze the average resource consumption of the process. Let $M$ denote the random variable of consumed copies of $\ket{\psi_{\phi}}$ in the described overall process of implementing $C_{\phi}$ and let $R$ denote the random variable counting the required rounds. Then
\begin{align}
E(M) &= \sum_{m=0}^\infty m \ p(m) = \sum_{m=0}^{\infty} m \sum_{r=1}^L p(r) p(m|r) \\
 &=  \sum_{r=1}^L  p(r)  \underbrace{\sum_{m=0}^{\infty} m \  p(m|r)}_{E_r(M)},
\end{align}
where $E(M)$ denotes the expectation value of $M$, $p(m|r)$ is the conditional probability of requiring $m$ copies of the magic state given $r$ rounds are required, and $E_r(M)$ is the conditional expectation value of $M$ given $r$ rounds are required. The average cost of generating a copy of $\ket{\psi_{2^k\phi}}$ is $4^k$ copies of $\ket{\psi_{\phi}}$. Hence, $E_r(M)=\sum_{k=1}^r 4^{k-1}$. Moreover, $p(r) = 
\begin{cases}
    1/2^r,& \text{if } 1 \leq r < L\\
    1/2^{L-1},              & \text{if } r=L
\end{cases}$. With this, a short calculation shows that $E(M) = 2^L -1$.

Let us now choose $L = \log 2/\epsilon$ leading to a probability of $1 - \epsilon/2$ of not requiring more than $L$ rounds. The expected resource consumption is then $E(M) = 2/\epsilon -1 \approx 2/\epsilon$. Markov's inequality, which gives a bound on the probability that a positive random variable takes values larger than a multiple of its expectation value, can be invoked to see that a supply of $4/\epsilon^2$ copies of $\ket{\psi_{\phi}}$ suffices with a probability of at least $1 - \epsilon/2$.  Hence, the probability of successfully implementing $C_\phi$ consuming no more than $4/\epsilon^2$ copies of $\ket{\psi_{\phi}}$ is at least $1-\epsilon$. This shows that $\ket{\psi_{\phi}}$ satisfies (M1) and completes the proof.
\end{proof}

We remark that a proof of an approximate version of Lemma \ref{lemma:4qubitmagicfamily} can be obtained from the theory of random walk hitting times. The process in the second paragraph of the proof above (implementing either $C_\phi$ or $C_{-\phi}$ with probabilities half) if simply iterated, amounts to a random walk on the circle of $\phi$ values, stepping left or right by angle $\phi$ in each step while consuming a single copy of $\ket{\psi_\phi}$. Then one may argue \cite{alpha} that for any $\epsilon >0$, ${\mathcal{O}}(\operatorname{poly} (1/\epsilon))$ steps will suffice to hit a value $\tilde{\phi}\in (\phi-\epsilon, \phi+\epsilon)$ with probability $1-\epsilon$, thereby resulting in implementation of a gate $C_{\tilde{\phi}}$ with $||C_{\tilde{\phi}} - C_{\phi}||<\epsilon$.

Let us remark here, that the magic states that we have considered before, in particular $\ket{M}$, are equivalent to $\ket{\psi_{\phi}}$ with $\phi = \pi$ up to transformation by matchgates. Together with the following lemma, Lemma \ref{lemma:4qubitmagicfamily} shows that every non-Gaussian, fermionic four-qubit state serves as a magic state in matchgate circuits.

\begin{lemma}
\label{lemma:equivalenc4qubits}
Any fermionic four-qubit state $\ket{\psi}$ is MG-equivalent to $\ket{\psi_{\phi}} = 1/2 (\ket{0000} + \ket{0011} + \ket{1100} + e^{i \phi} \ket{1111})$ for some $\phi \in [0, 2\pi]$.
\end{lemma}
\begin{proof}
First note that any odd fermionic four-qubit state can be turned into an even fermionic four-qubit state by applying $G(X,X)=X\otimes X$ to one of its qubits and an ancillary $\ket{0}$ state, so we can assume that $\ket{\psi}$ is even.
Defining $\ket{e_0} = \ket{00}$,  $\ket{e_1} = \ket{11}$, $\ket{d_0} = \ket{01}$,  and $\ket{d_1} = \ket{10}$, an arbitrary even fermionic four-qubit state $\ket{\psi}$ can be written as $\ket{\psi_e} + \ket{\psi_d}$ for unnormalized
 $\ket{\psi_e} \in \operatorname{span}\{\ket{e_0 e_0},\ket{e_0 e_1},\ket{e_1 e_0},\ket{e_1 e_1}\}$ ($e$-subspace) and $\ket{\psi_d} \in \operatorname{span}\{\ket{d_0 d_0},\ket{d_0 d_1},\ket{d_1 d_0},\ket{d_1 d_1}\}$  ($d$-subspace). We will now prove the lemma by constructing a depth-3 matchgate circuit that transforms $\ket{\psi}$  into the desired form, $\ket{\psi_{\phi}}$, for some $\phi \in [0,2 \pi]$.

A matchgate $G(A,B)$, with $\det{A} = \det{B}$, either on lines 1 and 2, or on lines 3 and 4 applies  the unitary $A$ ($B$) to the subspace spanned by $\{\ket{e_0},\ket{e_1}\}$ ($\{\ket{d_0},\ket{d_1}\}$), respectively. If one wished to apply an arbitrary unitary $A$ on the $\{\ket{e_0},\ket{e_1}\}$ subspace, one can achieve this by applying $G(A, e^{i \xi/2}\identity)$, where $e^{i \xi} = \det{A}$, which applies the desired unitary $A$ on the $e$-subspace, and only an overall phase on the $d$-subspace. Simliarly, $G(e^{i \eta/2}\identity,B)$ can be used to implement any unitary $B$ with $e^{i \eta} = \det{B}$ on the $d$-subspace.
This allows us to bring the state $\ket{\psi}$ into the form
\begin{align}
\ket{\psi'} = p \ket{e_0 e_0} + q \ket{e_1 e_1} + r \ket{d_0 d_0} + s \ket{d_1 d_1}
\end{align}
via matchgates $G(A_{12},e^{i \xi_{12}/2}\identity) \otimes G(A_{34},e^{i \xi_{34}/2}\identity)$, followed by $G(e^{i \eta_{12}/2},B_{12}) \otimes G(e^{i \eta_{34}/2},B_{34})$  that bring the $e$-subspace and the $d$-subspace independently into Schmidt-form. Note also that this is a depth-1 circuit.

Moreover, as $fSWAP_{2,3} $ maps $\ket{d_0d_0}= \ket{0101}$ to $\ket{0011}=\ket{e_0e_1}$, $\ket{d_1d_1}$ to $\ket{e_1e_0}$ and $\ket{e_je_j}$ to $(-1)^j\ket{e_je_j}$, we see that applying $fSWAP_{2,3}$ brings the state $\ket{\psi'}$ into the form
\begin{align}
\ket{\psi''} = p \ket{e_0 e_0} + q' \ket{e_0 e_1} + r' \ket{e_1 e_0} + s' \ket{e_1 e_1},
\end{align}
which is supported on the $e$-subspace solely, as is the target state $\ket{\psi_{\phi}}$.

Note that in the $e_0/e_1$-encoding, the Schmidt coefficients of the state $\ket{\psi_{\phi}}$ are $\lambda_{1,2} = \frac{1}{2} \pm \frac{1}{2} \cos(\phi/2)$. Thus, varying the angle $\phi$, all possible Schmidt values are covered. Hence, for all possible $\ket{\psi''}$, there exists a $\ket{\psi_{\phi}}$ having the same Schmidt values. We now use the fact, that two bipartite states are interconvertible into each other by local unitaries iff they have the same Schmidt coefficients. It implies that there always exists a pair of matchgates $G(A'_{12},e^{i \xi'_{12}/2}\identity) \otimes G(A'_{34},e^{i \xi'_{34}/2}\identity)$ which brings $\ket{\psi''}$ into the form $\ket{\psi_{\phi}}$ for some $\phi \in [0, 2 \pi]$ depending on the Schmidt values of $\ket{\psi''}$.

Concatenating the three steps, we see that any even fermionic state $\ket{\psi}$ can be transformed into some $\ket{\psi_{\phi}}$ by a depth-3 matchgate circuit, completing the proof of the lemma.
\end{proof}

Let us remark here that $\ket{\psi_{\phi}}$ is Gaussian iff $\phi = 0$ and the above lemmas have established that all non-Gaussian fermionic 4-qubit states are magic state for matchgate computations. Finally in the following, we show that actually all non-Gaussian fermionic $k$-qubit states are magic states for any $k\geq 4$.

\begin{theorem}
\label{Th_allfermi}
Any pure fermionic state which is non--Gaussian is a magic state for matchgate computations.
\end{theorem}
\begin{proof}
Having established that every fermionic non-Gaussian 4-qubit state is a magic state, we will prove the theorem inductively by showing that any non-Gaussian fermionic $k$-qubit state $\ket{\psi}$ with $k\geq 5$ can be transformed into a fermionic non-Gaussian $(k-1)$-qubit state by matchgates and adaptive measurements in the computational basis, which then (by reducing $k$ to 4) will imply that $\ket{\psi}$ is a magic state.
 Our procedure for obtaining a non-Gaussian state on fewer qubits  will be probabilistic. However, as the considered $k$-qubit states $\ket{\psi}$ do not depend on the circuit size $n$, the probability of success can be made arbitrarily close to 1 by repeating the procedure a number of times that does not depend on $n$. 
 
 To this end we make use of the operator $\Lambda_k = \sum_{i=1}^{2k} c_i \otimes c_i$ defined on two copies of the original Hilbert space. As mentioned above, an operator $X$ on $k$ qubits is Gaussian iff $\left[  \Lambda_k, X^{\otimes 2}  \right] = 0$ \cite{Br05Lagrange}. Moreover, a fermionic $k$-qubit state $\ket{\phi}$ is Gaussian iff $\Lambda_k \ket{\phi}^{\otimes 2} = 0$ \cite{SpSc18}.
Note that in particular, $\Lambda_k \ket{\psi}^{\otimes 2} \neq 0$, as $\ket{\psi}$ is non-Gaussian, $\left[  \Lambda, U_{\text{MG}}^{\otimes 2}  \right] = 0$, where $U_{\text{MG}}$ is any nearest-neighbor matchgate acting on parties $i$ and $i+1$ on both copies, and $\left[  \Lambda, \left(P_b^{(i)} \right)^{\otimes 2}  \right] = 0$, where $P_b^{(i)} = \identity \otimes \ldots \otimes \identity \otimes \ket{b}\bra{b} \otimes \identity \otimes \ldots \otimes \identity$, $b \in \{0,1\}$, and the projector onto $\ket{b}$ is acting on party $i \in \{1, \ldots, k\}$.

Abbreviating
\begin{align}
\label{eq:lambdaeq}
\ket{v} = \Lambda_k \ket{\psi}^{\otimes 2},
\end{align}
we distinguish now the following two cases. In case 1, there exists some $j \in \{1, \ldots, k\}$ and $b \in \{0,1\}$ such that  $\left(P_b^{(j)} \right)^{\otimes 2} \ket{v} \neq 0$, while in case 2, $\left(P_b^{(j)} \right)^{\otimes 2} \ket{v} = 0$ for all $j \in \{1, \ldots, k\}$ and $b \in \{0,1\}$. We will now show that it is possible to transform $\ket{\psi}$ into a non-Gaussian $(k-1)$-qubit state in both cases.

In case 1, we have that $\left(P_b^{(j)} \right)^{\otimes 2} \ket{v} = \Lambda_k  \left(P_b^{(j)} \right)^{\otimes 2} \ket{\psi}^{\otimes 2}  \neq 0$, for some $j$ and $b$. Hence, there exists some $j$ such that when measuring qubit $j$ in the computational basis, at least one of the two possible outcomes, outcome $b$, yields a non-Gaussian state (with non-vanishing probability). Qubit $j$ of that state separates from the remaining qubits and is in state $\ket{b}$. As such, we  can discard (swap away) line $j$ in order to successfully obtain a non-Gaussian state on $k-1$ qubits.

Let us now discuss case 2, which is a little more involved. In this case, $\ket{v}$ takes the form
\begin{align}
\ket{v} = \sum_{i \in \{0,1\}^k} \lambda_{i} \ket{i}\ket{\neg i},
\end{align}
with some coefficients $\lambda_{i}\in \mathbb{C}$, where $\neg i$ denotes bit-wise negation of $i$. (The state $\left|v \right>$ is actually symmetric under exchanging the first $k$ with the last $k$ qubits, hence, $\lambda_{i} = \lambda_{\neg i}$ but we will not utilize that symmetry). Note that in this case it is not possible to directly measure any of the qubits in the computational basis in order to obtain a non-Gaussian state on fewer qubits. However, we will show that the task can nevertheless be achieved by transforming $\ket{\psi}$ by a single matchgate $G(H,H)$ prior to measuring one of the qubits.
Let us distinguish the following two subcases. In subcase 2a, there exists $j \in \{1, \ldots, k-1\}$, $b \in \{0,1\}$ such that either $\left(P_b^{(j)} G(H,H)_{j,j+1} \right)^{\otimes 2} \ket{v} \neq 0$ or $\left(P_b^{(j+1)} G(H,H)_{j,j+1} \right)^{\otimes 2} \ket{v} \neq 0$. In subcase 2b, $\left(P_b^{(j)} G(H,H)_{j,j+1} \right)^{\otimes 2} \ket{v} = 0$ and $\left(P_b^{(j+1)} G(H,H)_{j,j+1} \right)^{\otimes 2} \ket{v} = 0$ for all $j$ and $b$. In case 2a it is possible to obtain a non-Gaussian, fermionic state on $k-1$ qubits by applying $G(H,H)_{j,j+1}$ onto $\ket{\psi}$ and then following the same procedure as in case 1. Let us now complete the proof of the theorem by showing that case 2b cannot occur. To do so, let us first rewrite the vector $\ket{v}$ as
\begin{align}
\ket{v} = \sum_{i \in \{0,1\}^{k-2}} &\lambda_{00 i} \ket{00 i} \ket{11 \neg i} +  \lambda_{01 i} \ket{01 i} \ket{10 \neg i} \nonumber\\
 &+  \lambda_{10 i} \ket{10 i} \ket{01 \neg i} +  \lambda_{11 i} \ket{11 i} \ket{00 \neg i}.
\end{align}
As $G(H,H)$ transforms the computational basis into the Bell basis, we have
\begin{align}
G(H,H)_{1,2}^{\otimes 2} \ket{v} =  & \nonumber\\
 1/2 \sum_{i \in \{0,1\}^{k-2}} & \ket{00 i} \ket{00 \neg i} (\lambda_{00 i} + \lambda_{11 i} )   \nonumber\\
                                                     + & \ket{11 i} \ket{11 \neg i}  (-\lambda_{00 i} - \lambda_{11 i} )    \nonumber\\
                                                     + & \ket{00 i} \ket{11 \neg i}  (-\lambda_{00 i} + \lambda_{11 i} )    \nonumber\\
                                                     + & \ket{11 i} \ket{00 \neg i}  (\lambda_{00 i} - \lambda_{11 i} )    \nonumber\\
                                                     + & \ket{01 i} \ket{01 \neg i}  (\lambda_{01 i} + \lambda_{10 i} )    \nonumber\\
                                                     + & \ket{10 i} \ket{10 \neg i}  (-\lambda_{01 i} - \lambda_{10 i} )    \nonumber\\
                                                     + & \ket{01 i} \ket{10 \neg i}  (-\lambda_{01 i} + \lambda_{10 i} )    \nonumber\\
                                                     + & \ket{10 i} \ket{01 \neg i}  (\lambda_{01 i} - \lambda_{10 i} ) 
\end{align}
Due to the assumption that every $\left(P_b^{(1)}\right)^{\otimes 2}$ or $\left(P_b^{(2)} \right)^{\otimes 2}$ maps $G(H,H)_{1,2}^{\otimes 2} \ket{v}$ to zero, $\ket{v}$ must obey the equations
\begin{align}
\label{eq:ghheqs}
 \lambda_{00 i} = -  \lambda_{11 i}, \  \lambda_{01 i} =  - \lambda_{10 i} \ \forall i \in \{0,1\}^{k-2}.
\end{align}
More generally, considering $G(H,H)$ acting on qubits $j$, $j+1$, one obtains that $\ket{v}$ must obey
$\lambda_{i} = -  \lambda_{i_1, \ldots, i_{j-1}, \neg i_j, \neg i_{j+1}, i_{j+2}, \ldots, i_k}$ for all $i \in \{0,1\}^{k}$ and $j\in \{1, \ldots, k-1\}$. We will now show that such a $\ket{v}$ cannot stem from a $\ket{\psi}$ as in Eq. (\ref{eq:lambdaeq}). To do so, we make use of the fact that for operators $A$, $B$, and states $\ket{\chi}$, $\ket{\psi}$, and $\ket{\phi}$, it holds that $A \otimes B \ket{\chi}\ket{\xi} = \ket{\phi}$ iff $A \ket{\chi}\bra{\xi^*} B^T = \sum \phi_{i,j} \ket{i}\bra{j}$. Hence, Eq. (\ref{eq:lambdaeq}) is equivalent to
\begin{align}
\label{eq:lambdaeqreshaped}
\sum_{l=1}^{2k}c_l \ket{\psi}\bra{\psi^*} c_l^T=v.
\end{align}
Here, $v$ is an anti-diagonal matrix with entries $\lambda_i$ on its anti-diagonal. According to the assumption $v \neq 0$ and due to the restrictions on the $\lambda_i$ as in  Eq. (\ref{eq:ghheqs}), at least $2^{k-1}$ entries of $v$ are non-vanishing, which implies that the rank of the matrix $v$ is at least $2^{k-1}$. However, the rank of the matrix on the left hand side of Eq. (\ref{eq:lambdaeqreshaped}) is at most $2k$ leading to a contradiction if $k \geq 5$, which is the case. This proves that case 2b cannot occur. This completes the proof as in the other cases it is possible to obtain a non-Gaussian, fermionic state on $k-1$ qubits.

\end{proof}

Let us remark here, that case 2b in the proof of Theorem \ref{Th_allfermi}, which cannot occur in case $k \geq 5$, indeed occurs if $k=4$. In this case, the rank of $v$ equals 8. This is in line with the fact that all fermionic three-qubit states are Gaussian and hence it cannot be possible to transform a non-Gaussian four-qubit state to a non-Gaussian three-qubit state.


\begin{thebibliography}{99}
\bibitem{footnote1} Note that sometimes distillation protocols can be used to transform multiple copies of a state which fails to fulfil (M1) or (M2) into a magic state \cite{He19b}. A simple example is the three-qubit GHZ state $\left| GHZ_3 \right>$, which does not fulfil (M2). However, two copies of $\left| GHZ_3 \right>$ can be transformed into one copy of $\left| GHZ_4 \right>$, which is a magic state, via free operations.
\bibitem{footnote2} Alternatively, the subsequent sequence of matchgates could be adaptively updated.
\end{thebibliography}
\end{document}